\documentclass{svproc}
\usepackage{dirtytalk} % correct English quotation marks
\usepackage{type1cm}        % activate if the above 3 fonts are
\usepackage{graphicx}        % standard LaTeX graphics tool
\usepackage[bottom]{footmisc}% places footnotes at page bottom
\usepackage{xcolor}
\usepackage{algorithm}
\usepackage{algpseudocode}
\usepackage{newtxmath}       % selects Times Roman as basic font
\usepackage{hyperref}
\usepackage[style=numeric,maxbibnames=99]{biblatex}

\usepackage{url}

\begin{document}
	\mainmatter              % start of a contribution
	\title{Robust Bayes Acts under Prior Perturbations: Contamination, Stability, and Selection Paths}
	\titlerunning{Robustness of Bayes-Acts}  % abbreviated title (for running head)
	%                                     also used for the TOC unless
	%                                     \toctitle is used
	%
	\author{Christoph Jansen$^1$ \& Georg Schollmeyer$^2$}
	\authorrunning{Christoph Jansen \& Georg Schollmeyer} % abbreviated author list (for running head)
	%
	%%%% list of authors for the TOC (use if author list has to be modified)
	% \tocauthor{Ivar Ekeland, Roger Temam, Jeffrey Dean, David Grove,
	% 	Craig Chambers, Kim B. Bruce, and Elisa Bertino}
	%
	\institute{$^1$School of Computing \& Communications, Lancaster University Leipzig, Germany\\
    $^2$ Department of Statistics, Ludwig-Maximilians-Universit\"at M\"unchen, Germany}
	
	\maketitle              % typeset the title of the contribution
	
	\begin{abstract}
This paper develops a quantitative framework to assess the robustness of Bayes-optimal decisions in finite decision problems under model uncertainty. We introduce two complementary stability notions for acts: the robustness radius, measuring the largest perturbation of a reference prior under which an act remains Bayes-optimal, and the contamination need, quantifying the minimal perturbation required for an act to become Bayes-optimal under some nearby prior. Both concepts are characterized via linear programming formulations and computed efficiently using bisection methods exploiting monotonicity properties. Building on these stability measures, we propose a cost-adjusted stability criterion that integrates robustness considerations with act-specific selection costs, yielding a parametric family of decision rules indexed by a regularization parameter. We analyze how optimal act selection evolves along this parameter and derive selection paths that reveal structural transitions between stability-driven and cost-driven regimes. The framework is applied to a portfolio choice problem under uncertainty between different economic regimes. Concretely, using data on historical ETF returns, we compute robustness and contamination profiles for six portfolio strategies and analyze their behavior under heterogeneous belief specifications. The results illustrate that robustness-based selection refines classical expected utility by accounting for prior misspecification.

	\end{abstract}
\section{Introduction}
Decision theory under uncertainty, as pioneered by the works of von Neumann and Morgenstern and Savage (\cite{vnm1944,s1954}), is built on the assumption that a decision maker evaluates acts by their expected utility with respect to a unique probability measure over states of nature. While this framework provides a powerful normative foundation for decision making under uncertainty, its practical applicability depends critically on the specification of a precise prior distribution modelling the decision maker's beliefs. In many realistic settings, however, such a precise belief is either unavailable or only weakly justified, leading to substantial sensitivity of optimal decisions to small perturbations in the prior (e.g., \cite{held2009,beer2013,JANSEN201849}). This sensitivity has motivated a broad literature on decision making under partial or imprecise probabilistic information (\cite{kofler,walley1991,augustin2014introduction}), including approaches based on credal sets (\cite{levi1974}), robust statistical models (\cite{huber1980}), and regularization techniques (\cite{hodges1952use}). These frameworks typically replace a single prior with a set of plausible distributions or modify the decision criterion to penalize undesirable behavior under worst-case scenarios. While effective, these approaches often shift the focus from the structure of Bayes-optimality itself to alternative notions of admissibility or robustness (e.g., \cite{jansen2018some,jsa2018,jbas2022,dependent,Jansen2023multi,jansen2025contributions}).
\\[.15cm]
In contrast, following a short introduction to classical decision theory (Section~\ref{dtu}), this paper proposes a direct robustness analysis of Bayes-optimal acts as a function of perturbations in the underlying prior. We introduce two complementary stability measures (Section~\ref{measuring}). The first, the \textit{robustness radius}, quantifies how far a reference prior can be perturbed while preserving the Bayes-optimality of a given act across all similar distributions. The second, the \textit{contamination need}, captures the minimal perturbation required to render an act Bayes-optimal under some admissible perturbation of the reference belief. Together, these measures provide a stability analysis of classical Bayes acts. A key feature of our approach is that both quantities admit exact characterizations via linear programming problems with perturbation constraints. As a result, robustness analysis becomes computationally tractable even in moderately high-dimensional decision problems. We introduce a cost-adjusted stability criterion that combines robustness measures with exogenous selection costs. This yields a parametric decision rule indexed by a regularization parameter, which governs the trade-off between stability with respect to prior misspecification and selection costs. The resulting selection paths provide a dynamic view of how optimal decisions evolve as the emphasis shifts from robustness to  cost efficiency. We showcase the usefulness of the proposed framework in a portfolio choice setting under uncertainty about different economic regimes (Section~\ref{portfolio}). 
%Using a finite set of economic states derived from empirical return-volatility clustering, we construct utility matrices based on historical ETF returns and specify heterogeneous priors reflecting different economic narratives. The analysis of robustness radii, contamination needs, and cost-adjusted selection paths reveals how different portfolio strategies vary in their sensitivity to belief perturbations and how robustness considerations interact with cost constraints in shaping optimal choices.
%
\section{Decision Theory under Uncertainty} \label{dtu}
\subsection{The Basic Setup}
Throughout, we consider the basic model of finite decision theory (e.g., \cite{.2005}). A \textit{decision maker} (or \textit{agent}) is faced with the task of choosing an act $a_i$ from a finite set $\mathbb{A}=\{a_1 , \dots , a_n\}$ of available alternatives. The consequence of this choice is uncertain, as the utility of an act depends on the realization of a \textit{state of nature} drawn from a finite set $\Theta=\{\theta_1 , \dots , \theta_m\}$. We assume that the utility associated with each pair $(a, \theta) \in \mathbb{A} \times \Theta$ is described by a real-valued \textit{cardinal utility function} $u:\mathbb{A} \times \Theta \to \mathbb{R}$, which is unique up to positive affine transformations. We write $u_{ij}:= u(a_i , \theta_j)$ for the utility obtained when act $a_i$ is chosen and state $\theta_j$ occurs.
\begin{table}
\caption{The basic model of finite decision theory.}
\begin{center}
\begin{tabular}{cccc}
\hline\noalign{\smallskip}
$u(a_i , \theta_j)$ & $\mathbf{\theta_1}$ & $\cdots$ & $\mathbf{\theta_m}$\\
\noalign{\smallskip}
\hline
\noalign{\smallskip}
$\mathbf{a_1}$ & $u(a_1 , \theta_1)$ &$\cdots$ & $u(a_1 , \theta_m)$  \\
$\vdots$ & $\vdots$ & $\cdots$ & $\vdots$  \\
$\mathbf{a_n}$ & $u(a_n , \theta_1)$ & $\cdots$ & $u(a_n , \theta_m)$  \\
\hline
\end{tabular}
\end{center}
\label{bm-table}
\end{table}
For every act $a \in \mathbb{A}$, the utility function $u$ naturally induces a random variable $u_a: (\Theta, 2^{\Theta}) \to \mathbb{R}$ defined by $u_a(\theta):= u(a, \theta)$ for all $\theta \in \Theta$.  The structure of this decision problem is illustrated in Table~\ref{bm-table}. Given a \textit{finite decision problem} $(\mathbb{A} , \Theta , u)$ of this form, the central task is to identify an optimal act $a^*\in \mathbb{A}$. In the classical framework, following \cite{vnm1944,s1954}, optimality is defined via expected utility maximization. This requires the decision maker to specify a probability measure $\pi$ on $(\Theta, 2^{\Theta})$ representing their beliefs about the occurrence of states of nature. The optimality criterion is then formally defined as follows:
\begin{definition}
Any act $a^* \in \mathbb{A}$ that receives maximal expected utility with respect to $\pi$, i.e.~for which $\mathbb{E}_{\pi}[u_{a^*}] \geq \mathbb{E}_{\pi}[u_{a}]$  for all other $a \in \mathbb{A}$, is called \textbf{Bayes-act with respect to $\pi$}. We denote by $\mathbb{A}_{\pi}$ the set of all such acts from $\mathbb{A}$.    
\end{definition}
\subsection{Challenges in the Choice of the Prior}
Although maximizing expected utility is intuitively appealing and theoretically well founded, it faces a key practical challenge: specifying the distribution $\pi$. While Savage's theorem \cite{s1954} guarantees the existence of such a distribution under suitable rationality axioms, this does not provide much guidance for its concrete determination. In practice, one often only has partial information, such as constraints on probabilities (e.g., $\pi(\theta_1) \leq \pi(\theta_2)$), rather than precise values (e.g., $\pi(\theta_1)=0.33$). Literature offers two main approaches to address this:
\\[.2cm]
\noindent \textbf{Credal set approach:} Following \cite{levi1974}, collect all probability measures on $(\Theta , 2^{\Theta})$ consistent with the given constraints in a \textit{credal set} $\mathcal{M}$. Define $\texttt{ag}: \mathbb{A} \to \mathbb{R}$ by $\texttt{ag}(a):=f(\{\mathbb{E}_{\pi}[u_{a}]: \pi \in \mathcal{M}\})$, where $f:2^{\mathbb{R}} \to \mathbb{R}$. Then, choose $a^* \in \text{argmax}_a~\texttt{ag}(a)$. Common choices include $f=\inf$ ($\Gamma$-minimax), $f=\sup$ ($\Gamma$-maximax), and $f=\eta \cdot \inf+(1-\eta)\cdot \sup$ ($\Gamma$-Maximix); see \cite{Troffaes.2007,jsa2018,festschrift}.
\\[.2cm]
\noindent \textbf{Regularization approach:} Fix a prior $\pi_0\in \mathcal{M}$ and a \textit{trust parameter} $0\leq \mu \leq 1$. Define $\texttt{rex}: \mathbb{A} \to \mathbb{R}$ by $\texttt{rex}(a):=\mu \cdot \mathbb{E}_{\pi}[u_{a}] +(1 - \mu) \cdot \min \{u(a, \theta): \theta \in \Theta\}$.  Then, choose $a^* \in \text{argmax}_a~\texttt{rex}(a)$. Originally proposed in \cite{hodges1952use} (with solution algorithms in \cite{meets}), this can be viewed as Tikhonov-type regularization, where the penalty term enforces a minimum utility level (see \cite[p.~397]{hodges1952use}).
\\[.2cm]
\noindent 
Both approaches are well established but have drawbacks. The credal set approach requires specifying an aggregation function $f$ modelling the decision maker's \textit{attitude towards ambiguity}, thereby introducing an additional modeling layer. Regularization depends on the choice of the trust parameter $\mu$, whose calibration directly influences the outcome and may lack a clear interpretation.
\\[.2cm]
\noindent \textbf{Stability-based robustness approach:} Rather than modifying the objective or replacing Bayesian reasoning, one can analyze the \emph{stability} of Bayes-optimal acts under perturbations of the prior. For a fixed $\pi_0$, we measure (i) how far $\pi_0$ can be perturbed while preserving optimality of $a$, and (ii) how much perturbation is needed for $a$ to become optimal. These notions, formalized below as robustness radius and contamination need, provide a sensitivity analysis within the EU framework. Unlike the credal set approach, this perspective avoids specifying a global ambiguity attitude. Unlike regularization, it does not alter the utility function. Instead, it evaluates how optimality behaves under controlled deviations from $\pi_0$, yielding a geometric view in which acts are compared by both expected utility and the size of the region of priors supporting their optimality.
\section{Measuring the Robustness of Bayes-Acts}\label{measuring}
In this section, we introduce a local stability analysis of Bayes-optimal decisions with respect to perturbations of the underlying prior distribution. The central idea is to quantify how sensitive the optimality of a given act is to small changes in probabilistic beliefs. This leads to two complementary notions: a robustness radius, which measures stability under perturbations of the prior, and a contamination need, which captures the minimal perturbation required for an act to become Bayes-optimal in some nearby model.
\\[.15cm]
We begin by formalizing the type of perturbations considered. Let $\Delta_{m-1}$ denote the probability simplex over the finite state space $(\Theta, 2^{\Theta})$, where $\Theta=\{\theta_1 , \dots , \theta_m\}$. For a fixed reference prior $\pi_0 \in \Delta_{m-1}$ and a perturbation level $\varepsilon \in [0,1]$, we define a neighborhood of admissible priors as follows:
\begin{equation*}
    B(\pi_0 , \varepsilon):= \Bigr\{\pi \in \Delta_{m-1}: \pi_0(\theta_j)-\varepsilon \leq \pi(\theta_j)\leq \pi_0(\theta_j)+\varepsilon \text{ for all }  j=1 , \dots,m\Bigl\}
\end{equation*}
This construction imposes a coordinate-wise $\ell_\infty$-type uncertainty band around the reference belief, allowing us to study robustness in a controlled and interpretable way. We now define the two central stability measures.
\begin{definition}\label{robustness}
Let $(\mathbb{A} , \Theta , u)$ be a finite decision problem, let $\pi_0 \in \Delta_{m-1}$, and let $a \in \mathbb{A}$ some fixed act. We define the following quantities:
\begin{itemize}\itemsep2mm
    \item[i)] The \textbf{robustness radius} of $a$ with respect to $\pi_0$ is defined by
    $$\emph{rob}(a,\pi_0):= \sup \big\{\varepsilon: a \in \mathbb{A}_{\pi} \emph{ for all } \pi \in  B(\pi_0 , \varepsilon)\big\},$$
    where we use the convention that $\sup \emptyset = -\infty$.
    \item[ii)] The \textbf{contamination need} of $a$ with respect to $\pi_0$ is defined by
    $$\emph{con}(a,\pi_0):= \inf \big\{\varepsilon: a \in \mathbb{A}_{\pi} \emph{ for some } \pi \in  B(\pi_0 , \varepsilon)\big\},$$
    where we use the convention that $\inf \emptyset =+ \infty$.
\end{itemize}
\end{definition}
Intuitively, the robustness radius captures a worst-case notion of stability: it is the largest perturbation radius within which the act remains optimal under all admissible belief perturbations. Therefore, the robustness radius can be understood as an adaption of the notion of a breakdown point known from robust statistics, (e.g., \cite{hampel,hodges,donoho_huber}) to the setting of a selector of an optimal decision in decision theory.  In contrast, the contamination need captures an existential notion: the smallest perturbation required so that the act becomes optimal for at least one nearby prior. These two quantities therefore provide complementary perspectives on the stability landscape of Bayes-optimality.

\begin{proposition}
In the situation of Definition~\ref{robustness} we have that:
\begin{itemize}
    \item[i)] $\text{rob}(a,\pi_0)=-\infty$ if and only if $a \notin \mathbb{A}_{\pi_0}$.
    \item[ii)]  $\text{con}(a,\pi_0)=+\infty$ if and only if $a$ is \textit{strictly inadmissible}, i.e., if there exist weights $(\beta_b)_{b \in \mathbb{A}\setminus \{a\}}\ge 0$ with $\sum\nolimits_{b \ne a}\beta_b=1$ such that $u(a,\theta_j) < \sum\nolimits_{b\ne a}\beta_b\cdot u(b,\theta_j)$ for all $j=1 , \dots , m$.
\end{itemize}
\end{proposition}
%
\iffalse
\begin{proof}
\textit{i)} is immediate. We show \textit{ii)}, i.e.,  $\mathrm{con}(a,\pi_0)=+\infty$ $\Rightarrow$ $a$ is inadmissible.
For this, suppose $\mathrm{con}(a,\pi_0)=+\infty$. By definition, there is no prior $\pi \in \Delta_{m-1}$ such that $a \in \mathbb{A}_\pi$. Hence, for every $\pi \in \Delta_{m-1}$ there exists $b \neq a$ with $\mathbb{E}_\pi[u_b] > \mathbb{E}_\pi[u_a]$, i.e.,
$\sum_{j=1}^m \pi(\theta_j)\bigl(u(a,\theta_j)-u(b,\theta_j)\bigr) < 0$. Define $d^{(b)} \in \mathbb{R}^m$ by $d^{(b)}_j := u(a,\theta_j)-u(b,\theta_j).$ 
Then for all $\pi \in \Delta_{m-1}$ we have $\min_{b \neq a} \langle \pi, d^{(b)} \rangle < 0$. By Farkas' lemma, this implies the existence of weights $(\beta_b)_{b \ne a}$ with $\beta_b \ge 0$ and $\sum_{b \ne a} \beta_b = 1$ such that $\sum_{b \ne a} \beta_b d^{(b)} \le 0$ componentwise, with strict inequality in at least one coordinate. Expanding yields
$u(a,\theta_j) \le \sum_{b \ne a} \beta_b u(b,\theta_j)$ for all $j$, with strict $<$ for some $j$, hence $a$ is inadmissible. \hfill $\square$
\end{proof}
\fi
%
\begin{proof}
\textit{i)} is immediate. We show \textit{ii)}.

\smallskip
\noindent
``$\Rightarrow$'': Suppose $\mathrm{con}(a,\pi_0)=+\infty$. Then there is no prior $\pi \in \Delta_{m-1}$ such that $a \in \mathbb{A}_\pi$. Writing $\mathbb{A}\setminus\{a\}=\{b_1,\dots,b_k\}$, define $D \in \mathbb{R}^{k\times m}$ by
$D_{ij}:=u(a,\theta_j)-u(b_i,\theta_j)$. Identifying $\pi \in \Delta_{m-1}$ with a vector $\pi \ge 0$ satisfying $\mathbf 1^\top \pi =1$, the condition $a \in \mathbb{A}_\pi$ is equivalent to $D\pi \ge 0$. Hence the linear equation system
$D\pi \ge 0,\pi \ge 0,\mathbf 1^\top \pi =1$ has no solution.
Introducing slack variables $s\ge0$, this is equivalent to infeasibility of the system
$D\pi-s=0,
\mathbf 1^\top \pi =1,
\pi \ge0,
s\ge0.$
By Farkas' lemma (e.g., \cite[Section 5.8.3]{b2004}), there therefore exist $\beta \le 0$ and $\lambda<0$ such that
$\beta^\top D+\lambda \mathbf 1^\top \ge0,$ and thus $\tilde\beta^\top D < 0$, where
$\tilde\beta:=-\beta /\sum_i \beta_i$. Thus, we have $\sum_i \tilde\beta_i=1$ and $\tilde \beta_i \ge0$ and  the inequality gives
$u(a,\theta_j)
<
\sum\nolimits_{i=1}^k \tilde\beta_i\,u(b_i,\theta_j)$ for all  $j=1,\dots,m$.
So, $a$ is strictly inadmissible.

\smallskip
\noindent
``$\Leftarrow$'': Suppose there exists  $(\beta_b)_{b\ne a}$ with $\beta_b\ge0$ and $\sum_{b\ne a}\beta_b=1$ such that
$
u(a,\theta_j)
<
\sum_{b\ne a}\beta_b\,u(b,\theta_j)$ for all $j$.
Multiplying by $\pi(\theta_j)$ and summing over $j$ yields $\mathbb E_\pi[u_a]
<
\sum\nolimits_{b\ne a}\beta_b\,\mathbb E_\pi[u_b]$
for all $\pi \in \Delta_{m-1}$. Hence for every $\pi$ there exists $b\ne a$ with $\mathbb E_\pi[u_b]>\mathbb E_\pi[u_a]$, so $a\notin\mathbb A_\pi$ for all $\pi$. Therefore $\mathrm{con}(a,\pi_0)=+\infty$.
\hfill$\square$
\end{proof}
We next show that both quantities are characterized via linear optimization. This is crucial for computation, as it allows us to exploit efficient solvers and monotonicity properties in the perturbation parameter. The proofs for the following Propositions \ref{V}, \ref{R}, and \ref{mono} are straightforward and therefore omitted.
\begin{proposition} \label{V}
Let $(\mathbb{A} , \Theta , u)$ be a finite decision problem, let $\pi_0 \in \Delta_{m-1}$, and let $a \in \mathbb{A}$ some fixed admissible act. Then, the contamination need of $a$ with respect to $\pi_0$ can be computed via the following linear program:
\begin{equation*}
con(a,\pi_0)=
\min_{\pi,\varepsilon}
\left\{
 \varepsilon :
\begin{array}{l}
\pi_j \ge 0 \quad \forall j \\[4pt]
\sum_{j=1}^m \pi_j = 1 \\[4pt]
\pi_{0j} - \varepsilon \le \pi_j \le \pi_{0j} + \varepsilon \quad \forall j \\[4pt]
\sum_{j=1}^m \pi_j \big(u(a,\theta_j) - u(b,\theta_j)\big) \ge 0, \quad \forall b \ne a
\end{array}
\right\}
\end{equation*}
%
%Then, the contamination need of $a$ with respect to $\pi_0$ can be computed via:
%\begin{equation*}
%\mathrm{con}(a,\pi_0)
%=
%\inf \left\{ \varepsilon \ge 0 : V(\varepsilon) = 1 \right\}
%\end{equation*} 
\end{proposition}
A similar LP characterization holds for the robustness radius.
\begin{proposition}\label{R}
Let $(\mathbb{A} , \Theta , u)$ be a finite decision problem, let $\pi_0 \in \Delta_{m-1}$, and let $a \ne b \in \mathbb{A}$ with $a \in \mathbb{A}_{\pi_0}$. For $\varepsilon \geq 0$, consider the LP:
\begin{equation*}
R_{a,b}(\varepsilon)
:=
\min_{\pi}
\left\{
\sum_{j=1}^m \pi_j (u(a,\theta_j)-u(b,\theta_j))
: \begin{array}{l}
\pi_j \ge 0 \quad \forall j \\[4pt]
\sum_{j=1}^m \pi_j = 1 \\[4pt]
\pi_{0j} - \varepsilon \le \pi_j \le \pi_{0j} + \varepsilon \quad \forall j 
\end{array}
\right\}
\end{equation*}
Then, setting $R(\varepsilon):=\min_{b \ne a} R_{a,b}(\varepsilon)$, the robustness radius of $a$ with respect to $\pi_0$ can be computed via:
\begin{equation*}
\mathrm{rob}(a,\pi_0)
=
\sup \left\{ \varepsilon \ge 0 : R(\varepsilon) \ge 0 \right\}
\end{equation*}
\end{proposition}
This formulation reveals an important structure:
\begin{proposition}\label{mono}
Under the same conditions as in Propositions \ref{R}, the function 
 $\varepsilon \mapsto R(\varepsilon)$ is non-increasing.
%
%\begin{itemize}
%\item[(i)] $V(\varepsilon)$ is non-decreasing with $0 \le V(\varepsilon) \le 1$.
%
%\item[(ii)] $R(\varepsilon)$ is non-increasing.
%\end{itemize}
%Thus, both $\mathrm{con}(a,\pi_0)$ and $\mathrm{rob}(a,\pi_0)$ can be computed via bisection on $\varepsilon$.
\end{proposition}
With this, for a given act, we can compute the contamination need directly via linear programming and the robustness radius via bi-section by making use of Propositions \ref{R} and \ref{mono}. Pseudo-code for this is given in the following scheme. 
%
%\begin{figure}
% \begin{algorithm}[H]
% \caption{Bisection algorithm for the contamination need}
% \label{alg:contamination_bisection}

% \begin{algorithmic}[1]
% \State Input: $(\mathbb{A},\Theta,u)$, prior $\pi_0$, strictly admissible act $a$, tolerance $\delta$
% \State Initialize $\varepsilon_L \gets 0$, $\varepsilon_U \gets 1$

% \While{$\varepsilon_U - \varepsilon_L > \delta$}
%     \State $\varepsilon \gets (\varepsilon_L + \varepsilon_U)/2$
%     \State Compute $V(\varepsilon)$ via LP
%     \If{$V(\varepsilon) = 1$}
%         \State $\varepsilon_U \gets \varepsilon$
%     \Else
%         \State $\varepsilon_L \gets \varepsilon$
%     \EndIf
% \EndWhile

% \State \Return $\varepsilon_L$
% \end{algorithmic}
% \end{algorithm}
%    \caption{Bisection algorithm for the contamination need of an admissible act $a$.}
 %   \label{algo1}
%\end{figure}
%
%\begin{figure}
\begin{algorithm}[H]
\caption{Bisection algorithm for the robustness radius}
\label{alg:robustness_bisection}

\begin{algorithmic}[1]
\State \textbf{Input:} decision problem $(\mathbb{A},\Theta,u)$, prior $\pi_0$, act $a \in \mathbb{A}_{\pi_0}$, tolerance $\delta > 0$
\State \textbf{Initialize:} $\varepsilon_L \leftarrow 0$, $\varepsilon_U \leftarrow 1$
\State \textbf{Define:} $R(\varepsilon)$ as in Proposition above

\While{$\varepsilon_U - \varepsilon_L > \delta$}
    \State $\varepsilon \leftarrow (\varepsilon_L + \varepsilon_U)/2$
    \State Compute $R(\varepsilon)$ by solving $R_{a,b}(\varepsilon)$ for all $b \ne a$
    \If{$R(\varepsilon) \ge 0$}
        \State $\varepsilon_L \leftarrow \varepsilon$
    \Else
        \State $\varepsilon_U \leftarrow \varepsilon$
    \EndIf
\EndWhile

\State \Return $\varepsilon_L$
\end{algorithmic}
\end{algorithm}
%    \caption{Bisection algorithm for the robustness radius of an $\pi_0$-optimal act $a$.}
 %   \label{algo1}
%\end{figure}
%
\noindent Finally, we combine robustness radius, contamination need and externally given \textit{selection costs} for each candidate act $a \in \mathbb{A}$. In this way, we arrive at a decision criterion that is able to trade-off between acts that are close to optimal but have lower selection costs and acts that are fragile in their optimality but have higher selection costs. Examples for externally given selection costs can range from complexity penalties in, e.g.,  LASSO-regression to variance/volatility penalties in problems of mathematical finance (see our later application in Section \ref{portfolio}). 
\begin{definition}
Let $(\mathbb{A},\Theta,u)$ be a finite decision problem, let $\pi_0 \in \Delta_{m-1}$, and let $c:\mathbb{A}\to \mathbb{R}_{\ge 0}$ be a bounded selection cost function. Fix a weight parameter $\lambda \ge 0$. Define the \textbf{cost-adjusted stability score} by
\[
S_{\lambda,\pi_0}(a)
:=
\begin{cases}
\frac{\mathrm{rob}(a,\pi_0)}{\max_{a' \in \mathbb{A}} \mathrm{rob}(a',\pi_0)} - \lambda \,\frac{c(a)}{\max_{a' \in \mathbb{A}} c(a')},
& \text{if } a \in \mathbb{A}_{\pi_0}, \\[6pt]
- \frac{\mathrm{con}(a,\pi_0)}{\max_{a' \in \mathbb{A}} \mathrm{con}(a',\pi_0)} - \lambda \,\frac{c(a)}{\max_{a' \in \mathbb{A}} c(a')},
& \text{if } a \notin \mathbb{A}_{\pi_0}
\end{cases}
\]
A \textbf{cost-adjusted stability optimal act} is any
\[
a^* \in \arg\max_{a \in \mathbb{A}} S_{\lambda,\pi_0}(a).
\]
\end{definition}
%
%This criterion induces a smooth trade-off between robustness to model misspecification and selection cost, and will be analyzed empirically in the portfolio application in the following section.
%
\section{Portfolio Selection under Model Uncertainty}\label{portfolio}

We illustrate the proposed cost-adjusted stability criterion using financial return data, interpreting the problem as decision making under uncertainty. Rather than relying on historical observations directly, we use them to approximate the payouts under a finite set of economically meaningful scenarios. The uncertainty is then between these scenarios.
\\[.15cm]
\noindent \textbf{Constructing the action space.}
We consider $|\mathbb{A}|=6$ portfolio strategies constructed from exchange-traded funds (ETF's) representing major asset classes, including equities, bonds, commodities, real estate, and international equity. The underlying dataset consists of monthly log-returns over the period 2015--2024. ETF price data are obtained from Yahoo Finance via the \texttt{quantmod} R package \cite{yahoofinance2024, quantmod}, and portfolio construction follows standard asset-pricing methodology in the spirit of canonical factor and portfolio-sorting approaches \cite{frenchdata}. The six portfolios are defined as fixed convex combinations of the underlying ETF assets, representing heterogeneous investment strategies with distinct risk and return profiles. Specifically, we consider (i) an equity-core portfolio, (ii) a balanced equity portfolio, (iii) a diversified multi-asset portfolio, (iv) a bond-dominant portfolio, (v) a real-asset–tilted portfolio, and (vi) an equally weighted portfolio. The corresponding portfolio weights are reported in Table~\ref{tab:portfolio_weights}. %This construction ensures systematic variation in exposure across major asset classes while maintaining interpretability and full comparability across acts.
\\[.15cm]
\noindent \textbf{Constructing the state space.} To obtain a tractable representation of economic uncertainty, we define the state space as a finite set of four regimes:
\[
\Theta = \{\theta_1,\theta_2,\theta_3,\theta_4\}
=
\{\text{Expansion}, \text{Recession}, \text{Stagnation}, \text{Recovery}\}.
\]
The returns under these regimes are approximated from the historical return data using $k$-means clustering in the joint space of broad market performance and risk. Specifically, for each month $t$, we compute two features: (i) the market return $r_t^{M}$, proxied by the return on the S\&P 500 ETF (SPY), and (ii) a measure of market uncertainty given by realized volatility $\sigma_t$, computed as the rolling standard deviation of daily SPY returns within month $t$. Each monthly observation is therefore represented by a two-dimensional feature vector $(r_t^{M}, \sigma_t)$. We then apply $k$-means with $k=4$ these data to partition observations into four groups. The number of clusters is fixed to $k=4$ to align with the economic interpretation of the state space. Each cluster is interpreted ex post by inspecting its centroid in the $(r_t^{M}, \sigma_t)$ space:
clusters with high average returns and low volatility are labeled as expansions, clusters with low or negative returns and high volatility are labeled as recessions, and intermediate clusters are interpreted as stagnation or recovery depending on the relative ordering of returns and volatility. This induces a partition $\{T_1,\dots,T_4\}$ of the sample, where each $T_j$ corresponds to the set of time periods assigned to cluster $j$. Each observation is assigned to the regime whose centroid is closest in Euclidean distance. 
%This mapping transforms the continuous return-volatility process into a finite state representation of macro-financial conditions, which serves as the basis for the state space in the decision model.
%
\begin{table}[h!]
\centering
\caption{Portfolio construction via ETF weight allocations}
\label{tab:portfolio_weights}
\begin{tabular}{lccccc}
\hline
 & Equity & Bonds & Commodities & Real Estate & Intl. Equity \\
\hline
$a_1$ (Equity Core)        & 0.70 & 0.10 & 0.05 & 0.10 & 0.05 \\
$a_2$ (Balanced Equity)    & 0.50 & 0.20 & 0.05 & 0.15 & 0.10 \\
$a_3$ (Multi-Asset)        & 0.30 & 0.30 & 0.10 & 0.20 & 0.10 \\
$a_4$ (Bond Dominant)      & 0.15 & 0.70 & 0.05 & 0.05 & 0.05 \\
$a_5$ (Real-Asset Tilt)    & 0.25 & 0.15 & 0.25 & 0.30 & 0.05 \\
$a_6$ (Equal-Weight)       & 0.20 & 0.20 & 0.20 & 0.20 & 0.20 \\
\hline
\end{tabular}
\end{table}
\\[.15cm]
\noindent \textbf{Constructing the utility function.}
The utility of portfolio $a \in \mathbb{A}$ in state $\theta_j$ is defined as the conditional mean return:
\[
u(a,\theta_j)
:=
\tfrac{1}{|T_j|}
\sum\nolimits_{t \in T_j} r_t(a),
\]
where $T_j$ denotes the set of time periods classified into regime $\theta_j$, and $r_t(a)$ is the realized return of portfolio $a$ at time $t$. This yields a utility matrix $u \in \mathbb{R}^{6 \times 4}$, summarizing portfolio performance across macro-financial environments rather than individual time periods, see Table~\ref{tab:utility_matrix}.
\begin{table}[h!]
\centering
\caption{Utility matrix $u(a,\theta)$ across economic regimes}
\label{tab:utility_matrix}
\begin{tabular}{lcccc}
\hline
 & Expansion & Recovery & Stagnation & Recession \\
\hline
$a_1$ & 0.021 & 0.059 & -0.011 & -0.049 \\
$a_2$ & 0.018 & 0.051 & -0.011 & -0.037 \\
$a_3$ & 0.015 & 0.051 & -0.011 & -0.029 \\
$a_4$ & 0.006 & 0.024 & -0.007 & -0.024 \\
$a_5$ & 0.014 & 0.042 & -0.008 & -0.045 \\
$a_6$ & 0.015 & 0.041 & -0.009 & -0.024 \\
\hline
\end{tabular}
\end{table}
\\[.15cm]
\noindent \textbf{Specification of prior beliefs.}
To capture heterogeneous and economically interpretable beliefs about the four considered economic regimes, we consider a collection of eight prior distributions on the state space $\Theta$.
Each prior \(\pi_0 \in \Delta_3\) represents a stylized belief about the likelihood of regimes and is constructed to reflect distinct economic narratives. First, we include four \emph{regime-focused} priors that place the largest mass on a single state while maintaining strictly positive probability on all others: a boom-biased prior (Expansion), a recession-focused prior (Recession), a stagnation-oriented prior (Stagnation), and a recovery-focused prior (Recovery). Second, we consider four \emph{balanced and mixed} specifications: a uniform prior representing complete agnosticism, a growth-tilted prior that moderately favors expansions, a defensive prior emphasizing recession risk, and an opportunistic prior assigning weight to both expansion and stagnation scenarios. Figure~\ref{robvscont} visualizes the contamination need and the robustness radius of the different portfolios under all prior specifications.
\\[.15cm]
\noindent \textbf{Cost-adjusted stability acts and selection paths.}
Each portfolio is assigned a cost reflecting its risk exposure. For $a \in \mathbb{A}$, we define $c(a) := \mathrm{Var}(u_a)$, i.e.,
the variance of returns across regimes. Costs are normalized to lie in $[0,1]$, ensuring comparability across assets. We evaluate portfolios using the cost-adjusted stability score $S_{\lambda,\pi_0}(a)$, over a grid of parameters $0 \le\lambda \le 3$. For each $\lambda$, the selected portfolio is $a^* \in \arg\max_{a \in \mathbb{A}} S_{\lambda,\pi_0}(a)$. The selection paths for the portfolios along the different choices of $0 \le\lambda \le 3$ are visualized in Figure~\ref{selection}.
\begin{figure}
    \centering
    \includegraphics[width=\linewidth]{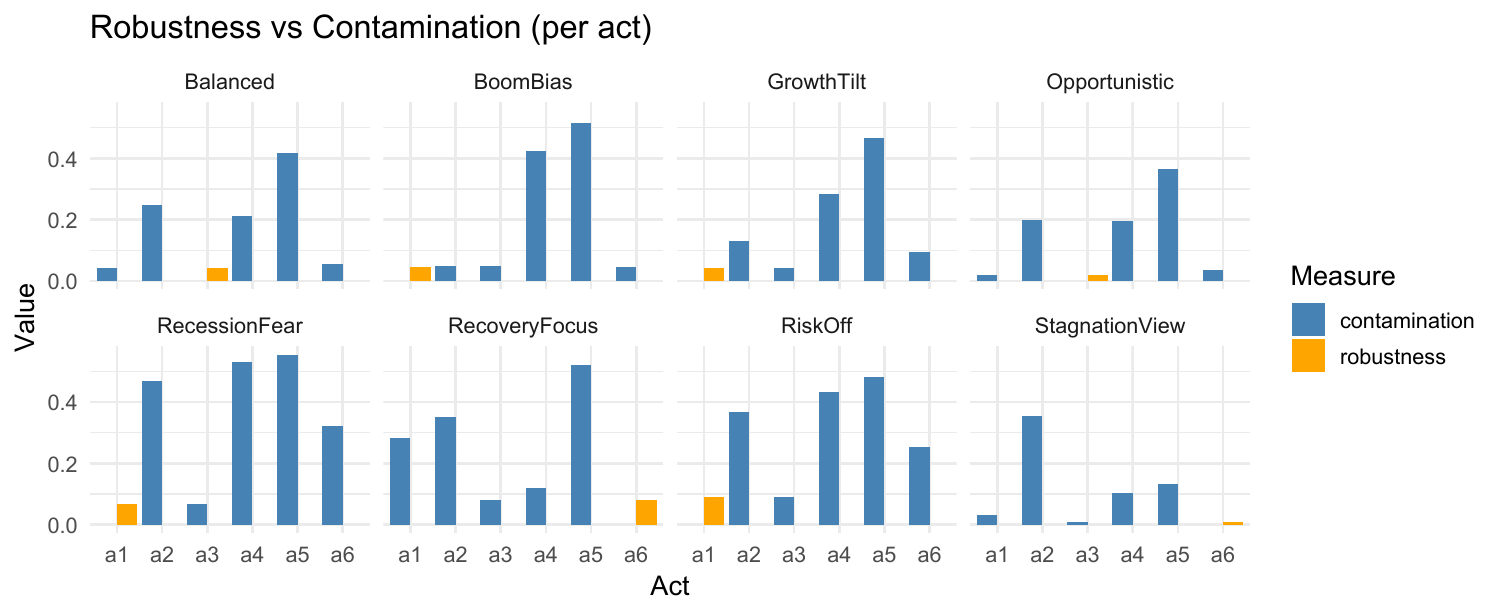}
    \caption{Contamination need and the robustness radius of the different portfolios $a_1 , \dots , a_6$ under all prior specifications.}
    \label{robvscont}
\end{figure}
\begin{figure}
    \centering
    \includegraphics[width=\linewidth]{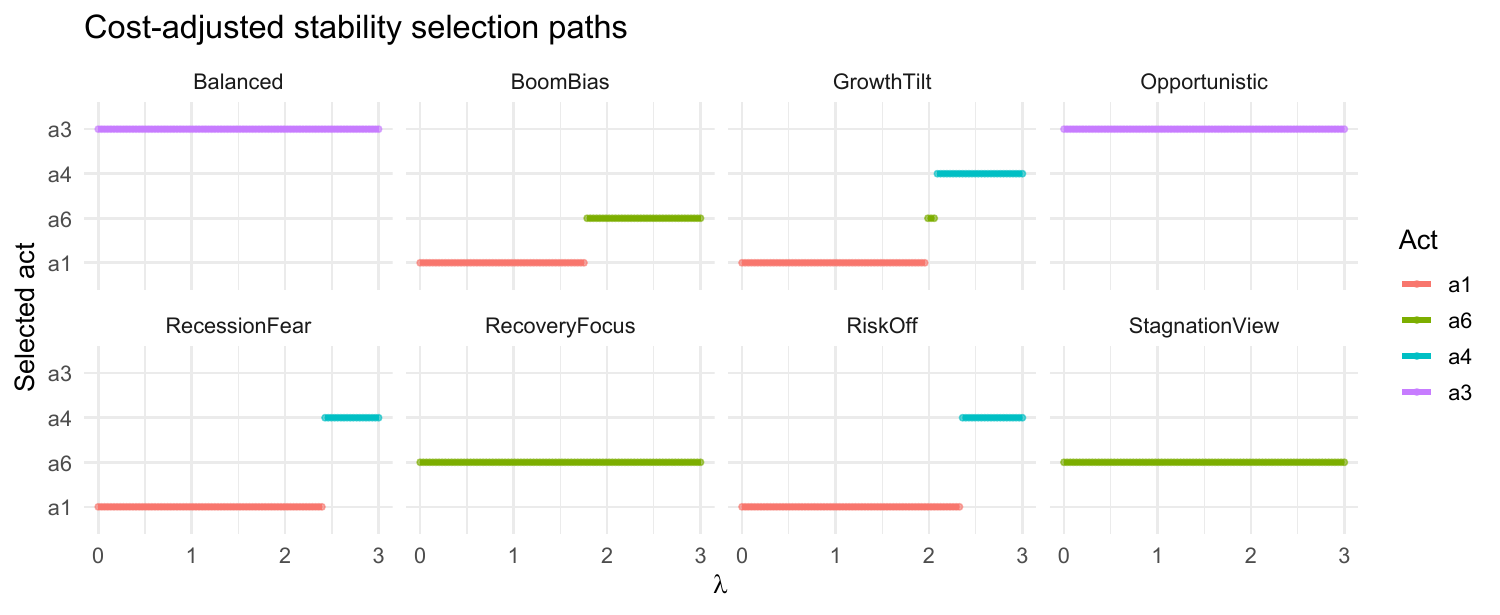}
    \caption{Selection paths of the cost-adjusted stability criterion for the portfolios $a_1 , \dots , a_6$ along the different choices of $0 \le\lambda \le 3$.}
    \label{selection}
\end{figure}
\\[.15cm]
\noindent \textbf{Interpretation of the results.} The results in Figure~\ref{robvscont} illustrate the complementary roles of contamination need and robustness radius: whenever the robustness radius is zero, the contamination need is strictly different from it. Portfolios such as the real-asset tilt portfolio $a_5$ and the equal-weight portfolio $a_6$ typically exhibit relatively large contamination needs (blue bars) across most priors, indicating that they remain suboptimal even under comparatively large perturbations of the underlying prior. In contrast, the equity-heavy portfolio $a_1$ is Bayes-optimal with respect to four out of eight priors, and has comparably high robustness for those priors (e.g., for the risk off prior up to $\approx 9\%$ probability mass could be switched between the states). The selection paths in Figure~\ref{selection} further show how the cost-adjusted stability criterion balances robustness radius and contamination need against external costs (here, the portfolios' variances). For small values of $\lambda$, portfolios with favorable robustness characteristics are preferred, whereas larger values of $\lambda$ increasingly favor portfolios with lower variance-based costs. Overall, the results demonstrate that robustness-based analysis refines classical expected utility comparisons by incorporating the stability of optimality under changing prior assumptions.

\section{Discussion and Future Work} 
The present work introduces \textit{robustness radius} and \textit{contamination need} as two complementary quantitative measures for evaluating the stability of Bayes-acts under perturbations of the underlying prior. Together with the \textit{cost-adjusted stability criterion}, these concepts provide a structured way to move beyond binary notions of optimality and instead describe how\textit{ fragile} or \textit{persistent} an optimal decision is when beliefs about the state space are subject to misspecification.
\\[.2cm]
\noindent
The empirical illustration in the portfolio selection setting suggests that these measures are not only theoretically well-defined but also practically meaningful. In particular, acts that are close in expected utility terms may exhibit markedly different robustness profiles, highlighting that classical Bayes-optimality alone is insufficient to capture decision stability. The introduction of selection costs further refines this picture by incorporating implementation considerations, thereby connecting classical decision theory with cost-of-action trade-offs. 
\\[.2cm]
\noindent
Several avenues for future work arise naturally from this framework. First, it would be of interest to extend the analysis from finite state spaces to general measurable spaces, where linear programming representations are no longer directly available. Second, the current notion of $\varepsilon$-contamination could be generalized to alternative ambiguity neighborhoods, such as Wasserstein balls, which may better reflect specific forms of model uncertainty. Third, the relationship between robustness radius, contamination need, and classical notions of admissibility and complete classes warrants further theoretical investigation. In particular, understanding whether these measures induce a meaningful ordering on the space of acts, or whether they can be embedded into a decision-theoretic representation theorem, remains an open question. Finally, from a machine learning perspective, the proposed stability measures may provide a principled way of quantifying prediction robustness in semi-supervised or self-training settings (e.g., \cite{rodemann,stefan}), where pseudo-label selection inherently induces decision-making under distributional uncertainty. Moreover, also recent decision theoretic-analyses of multicriteria benchmarking problems (e.g., \cite{uaiall,JMLR,jansen2024statistical,BLOCHER2024,gordienko2026}) might benefit from the proposed approaches for sensitivity analyses a lot. Bridging these areas appears to be a promising direction for future research.
\printbibliography
\end{document}